\documentclass[runningheads,a4paper]{llncs}
\pdfoutput=1 

\usepackage[utf8]{inputenc}
\usepackage[T1]{fontenc}

\usepackage{lmodern}
\usepackage{microtype}

\usepackage{enumitem}

\usepackage{stmaryrd}

\usepackage{amsmath}
\usepackage{amssymb}

\usepackage{tikz}

\usepackage{hyperref}

\usetikzlibrary{shapes}
\usetikzlibrary{positioning,calc}

\newcommand*{\var}[1]{\mathord{\mathit{#1}}}

\newcommand*{\fun}[1]{\mathord{\mathit{#1}}}
\newcommand*{\constr}[1]{\mathord{\mathit{#1}}}

\newcommand{\consts}{\mathcal{C}}
\newcommand{\vars}{\mathcal{V}}
\newcommand{\test}[3]{\texttt{=}(\var{#1},\var{#2},\var{#3})}

\newcommand{\actionreq}[3]{\texttt{+}(\var{#1},\var{#2},\var{#3})}
\newcommand{\actionmod}[3]{\texttt{=}(\var{#1},\var{#2},\var{#3})}

\newcommand{\actions}{\mathcal{A}}
\newcommand{\states}{\mathcal{S}}

\newcommand{\utrans}{\rightarrowtail}

\newcommand{\rules}{\Sigma}

\newcommand{\addinfo}{\upsilon}

\newcommand{\types}{\mathbb{T}}

\newcommand{\buffers}{\mathbb{B}}

\newcommand{\val}{\fun{val}}
\newcommand{\id}{\fun{id}}

\newcommand{\inv}{\mathcal{I}}
\newcommand{\actrinv}{\mathcal{A}}

\begin{document}

\title{A Decidable Confluence Test for Cognitive Models in ACT-R}
%
%
\author{Daniel Gall \and Thom Frühwirth}
\authorrunning{D. Gall \and T. Frühwirth} 
\institute{Institute of Software Engineering and Programming Languages, Ulm University, 89069~Ulm, Germany,\\
\email{\{daniel.gall,thom.fruehwirth\}@uni-ulm.de}}

\maketitle              

\begin{abstract}
Computational cognitive modeling investigates human cognition by building detailed computational models for cognitive processes. Adaptive Control of Thought -- Rational (ACT-R) is a rule-based cognitive architecture that offers a widely employed framework to build such models. There is a sound and complete embedding of ACT-R in Constraint Handling Rules (CHR). Therefore analysis techniques from CHR can be used to reason about computational properties of ACT-R models. For example, confluence is the property that a program yields the same result for the same input regardless of the rules that are applied.

In ACT-R models, there are often cognitive processes that should always yield the same result while others e.g. implement strategies to solve a problem that could yield different results.
In this paper, a decidable confluence criterion for ACT-R is presented.
It allows to identify ACT-R rules that are not confluent.
Thereby, the modeler can check if his model has the desired behavior.

The sound and complete translation of ACT-R to CHR from prior work is used to come up with a suitable invariant-based confluence criterion from the CHR literature. Proper invariants for translated ACT-R models are identified and proven to be decidable. The presented method coincides with confluence of the original ACT-R models. 

\end{abstract}

\section{Introduction}

Computational cognitive modeling is a research field at the interface of cognitive sciences and computer science. It tries to explain human cognition by building detailed computational models of cognitive processes \cite{sun_introduction_2008}. To support the modeling process, cognitive architectures like \emph{Adaptive Control of Thought -- Rational (ACT-R)} provide the ability to create models of specific cognitive tasks by offering representational formats together with reasoning and learning mechanisms to facilitate modeling \cite{taatgen_modeling_2006}. 

ACT-R is widely employed in the field of computational cognitive modeling. It is defined as a production rule system that offers advanced conflict resolution mechanisms to model learning and competition of different strategies for problem solving. Therefore, many ACT-R models are highly non-deterministic to resemble the applicability of more than one strategy in many situations. The strategy is chosen depending on information learned from situations in the past.

\emph{Confluence} is the property of a program that regardless of the order its rules are applied, they finally yield the same result. By identifying the rules that lead to non-confluence, model quality can be improved: It allows to check if the model has the desired behavior regarding competing strategies and e.g. identify rules that interfere with each other unintentionally.

In this paper, we present a \emph{decidable confluence test for the abstract operational semantics of ACT-R} using confluence analysis tools for CHR. In prior work, we presented a sound and complete embedding of ACT-R in CHR \cite{gall_ruleml2016,gall_tocl_2017}. An \emph{invariant-based confluence test for CHR} \cite{duck_observable_2007,raiser_phdthesis10} is used to decide confluence of the translated models with invariants on CHR states that come from the abstract operational semantics of ACT-R. The confluence test identifies the rules that lead to non-confluence supporting the decision if a model has the desired behavior regarding competing strategies.

First the preliminaries are recapitulated in section~\ref{sec:preliminaries}. The main section~\ref{sec:confluence_criterion} describes the confluence criterion for ACT-R models. For this purpose, the invariant-based confluence test for CHR is introduced briefly (section~\ref{sec:confluence_criterion:invariant_confluence}). Then, the ACT-R invariant is defined and a \emph{decidable criterion} for the invariant is given (section~\ref{sec:confluence_criterion:actr_invariant}). It is shown that the ACT-R invariant is maintained in the translation. The theoretical foundations to apply the CHR invariant-based confluence test to ACT-R models are derived resulting in a \emph{confluence criterion for terminating ACT-R models} (section~\ref{sec:confluence_criterion:invariant_confluence_actr}). An example is given in section~\ref{sec:example:counting}. 

\section{Preliminaries}
\label{sec:preliminaries}

\subsection{Confluence}

Confluence is the property of a state transition system that same inputs yield the same results regardless of which rules are applied. 
\begin{definition}[joinability and confluence \cite{fru_chr_book_2009}]
In a state transition system $(\states,\mapsto)$ with states $\states$ and a transition relation $\mapsto : \states \times \states$ with reflexive transitive closure $\mapsto^*$, two states $\sigma_1, \sigma_2 \in \states$ are \emph{joinable}, denoted as $\sigma_1 \downarrow \sigma_2$, if there exists a state $\sigma'$ such that $\sigma_1 \mapsto^* \sigma'$ and $\sigma_2 \mapsto^* \sigma'$.  A state transition system is \emph{confluent}, if for all states $\sigma, \sigma_1, \sigma_2: (\sigma \mapsto^* \sigma_1) \land (\sigma \mapsto^* \sigma_2) \rightarrow (\sigma_1 \downarrow \sigma_2).$

\end{definition}
Hence, a program is confluent if for all states that lead to different successor states, those states are joinable. A program is \emph{locally confluent}, if $(\sigma \mapsto \sigma_1) \land (\sigma \mapsto \sigma_2)$ in one transition step and $\sigma_1$ and $\sigma_2$ are joinable. It can be shown that for all state transition systems local confluence and confluence are equivalent \cite{fru_chr_book_2009}. Figure~\ref{fig:confluence} illustrates (local) confluence. 

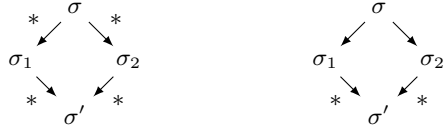
\begin{figure}[hbt]
 \centering
  \begin{tikzpicture}
\node (s) {$\sigma$};
\node[below left of=s] (s1) {$\sigma_1$};
\node[below right of=s] (s2) {$\sigma_2$};
\node[below right of=s1] (sp) {$\sigma'$};

\path[draw,-latex] (s) -- node[above left] {$*$} (s1);
\path[draw,-latex] (s) -- node[above right] {$*$} (s2);
\path[draw,-latex] (s1) -- node[below left] {$*$} (sp);
\path[draw,-latex] (s2) -- node[below right] {$*$} (sp);

\node[right of=s, node distance=4cm] (q) {$\sigma$};
\node[below left of=q] (q1) {$\sigma_1$};
\node[below right of=q] (q2) {$\sigma_2$};
\node[below right of=q1] (qp) {$\sigma'$};

\path[draw,-latex] (q) --  (q1);
\path[draw,-latex] (q) --  (q2);
\path[draw,-latex] (q1) -- node[below left] {$*$} (qp);
\path[draw,-latex] (q2) -- node[below right] {$*$} (qp);
\end{tikzpicture}

\caption{Confluence and local confluence.}
\label{fig:confluence}
\end{figure}

\subsection{Adaptive Control of Thought -- Rational (ACT-R)}

In this section, ACT-R is introduced briefly. An extensive introduction to the theory can be found in \cite{anderson_integrated_2004,taatgen_modeling_2006}.
ACT-R is a modular production rule system. Its data elements are so-called \emph{chunks}. A chunk has a \emph{type} and a set of \emph{slots} (determined by the type) that are connected to other chunks. Hence, human declarative knowledge is represented in ACT-R as a network of chunks. Figure~\ref{fig:chunk_count_order} shows an example chunk network that models the representation of an order over natural numbers.

\begin{figure}[htb]
\centering
\tikzstyle{chunk} = [circle, draw, text centered, text width=1em]
\tikzstyle{slot} = [draw, -latex]   

\begin{tikzpicture}[node distance = 2cm, auto]
 \node[chunk] (b) {$a$}; 
 \node[chunk, left of=b] (b-first) {1}; 
 \node[chunk, right of=b] (b-second) {2}; 

 \node[chunk, right of=b-second] (c) {$b$}; 
 \node[chunk, right of=c] (c-second) {3}; 

 \path[slot] (b) -- node {first} (b-first);
 \path[slot] (b) -- node {second} (b-second);
 \path[slot] (c) -- node {first} (b-second);
 \path[slot] (c) -- node {second} (c-second);
\end{tikzpicture}
\caption{A chunk network that represents the order of natural numbers $1, 2, 3$. The chunks are represented by nodes, the slots by labeled edges. The labels of the nodes are chunk identifiers. Chunks $1, 2$ and $3$ are of type \emph{number} that has no slots. Chunks $a$ and $b$ are of type \emph{order} that has a \emph{first} and a \emph{second} slot.}
\label{fig:chunk_count_order}
\end{figure}
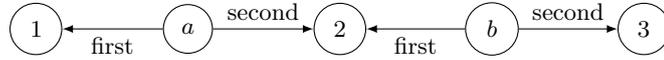

ACT-R's modules are responsible for different cognitive features. For instance, the declarative knowledge (represented as a chunk network) can be found in the \emph{declarative module}. Each module has a set of associated \emph{buffers} that contain at most one chunk. The heart of ACT-R is the \emph{procedural system} that consists of a set of \emph{production rules}. Those rules only have access to the contents of the buffers. They match the contents of the buffer, i.e. they check if the chunks of particular buffers have certain values. If a rule is applicable, it can \emph{modify} particular slots of the chunk in the buffer, \emph{request} the module to put a whole new chunk in its buffer or \emph{clear} a buffer. Modifications and clearings are available directly for the production rule system, whereas requests can take some time while the procedural system is continuing work in parallel.

\subsubsection{Syntax of ACT-R}

We use our simplified syntax in form of first-order terms that can be derived directly from the original syntax \cite{gall_ruleml2016,gall_tocl_2017}.
The syntax of ACT-R is defined over two disjoint sets of constant symbols $\consts$ and variable symbols $\vars$. An ACT-R model consists of a set of types $\types$ with type definitions and a set of rules $\rules$. 

A production rule has the form $\mathcal{L} \Rightarrow R$ where $\mathcal{L}$ is a finite set of buffer tests. A buffer test is a first-order term of the form $\test{b}{t}{P}$ where the buffer $b \in \consts$, the type $t \in \consts$ and $P \subseteq \consts \times (\consts \cup \vars)$ is a set of slot-value pairs $(s,v)$ where $s \in \consts$ and $v \in \consts \cup \vars$. This means that only the values in the slot-value pairs can consist of both constants and variables. 

The right-hand side $R \subseteq \actions$ of a rule is a finite set of actions where $\actions = \{ a(b,t,P) \enspace | \enspace a \in A, b \in \consts, t \in \consts \mbox{ and } P \subseteq \consts \times (\consts \cup \vars) \}$. Hence, an action is a term of the form $a(b,t,P)$ where the functor $a$ of the action is in $A$, the set of action symbols, the first argument $b$ is a constant (denoting a buffer), the second argument is a constant $t$ denoting a type, and the last argument is a set of slot-value pairs, i.e. a pair of a constant and a constant or variable. Usually, the action symbols are defined as $A := \{\mathtt{=},\mathtt{+},\mathtt{-}\}$ for modifications, requests and clearings respectively. Only one action per buffer is allowed, i.e. if $a(b,t,P) \in R$ and $a'(b',t',P') \in R$, then $b \neq b'$ \cite{actr_reference}. 

We assume the rules to be in so-called \emph{set normal form} that requires the slot tests of a rule to be total and unique with respect to the type of the test. This means that each slot defined by the type of the tested chunk must appear at most once in the set of slot-value pairs. Every rule can be transformed to set normal form \cite{gall_tocl_2017}.

\subsubsection{Operational Semantics of ACT-R}

For the understanding of this paper, it is sufficient to define ACT-R states and rules formally. The formal definition of the operational semantics can be found in \cite{gall_ruleml2016,gall_tocl_2017}. We define the operational semantics of ACT-R through our CHR translation that has been first presented in \cite{gall_ruleml2016} and in its most current form in \cite{gall_tocl_2017}. Since the translation is sound and complete, we omit the formal definition of the ACT-R semantics here, since it would only distract from the contribution of this paper.

\begin{definition}[chunk types, chunk stores]
A \emph{typing function} $\tau : \types \rightarrow 2^\consts$ maps each type from the set $\types \subseteq \consts$ to a finite set of allowed slot names. A \emph{chunk store} $\Delta$ is a multi-set of tuples $(t,\val)$ where $t \in \types$ is a chunk type and $\val : \tau(t) \rightarrow \Delta$ is a function that maps each slot of the chunk (determined by the type $t$) to another chunk. Each chunk store $\Delta$ has a bijective \emph{identifier function} $\id_\Delta : \Delta \rightarrow \consts$ that maps each chunk of the multi-set a unique identifier. 
\end{definition}

\emph{Additional information} represents the inner state of the modules and so-called sub-symbolic information used in ACT-R implementations to model cognitive features like forgetting, latencies and conflict resolution. The information is expressed as a conjunction of predicates from first-order logic. We now define ACT-R states as follows:
\begin{definition}[cognitive state,  ACT-R state]
\label{def:unified_state}
A \emph{cognitive state} $\gamma$ is a function $\buffers \rightarrow \Delta \times \mathbb{R}^+_0$ that maps each buffer to a chunk and a delay. 
The delay decides at which point in time the chunk in the buffer is available to the production system. A delay $d > 0$ indicates that the chunk is not yet available to the production system. This implements delays of the processing of requests.  

An \emph{ACT-R state} is a tuple $\langle \Delta; 
\gamma; \addinfo \rangle$ where 
$\gamma$ is a cognitive state and $\addinfo$ is a multi-set of ground, atomic first order predicates (called \emph{additional information}).
\end{definition}

\subsection{Constraint Handling Rules (CHR)}
\label{sec:constraint_handling_rules}

In this section, syntax and semantics of CHR are summarized briefly. For an extensive introduction to CHR, its semantics, analysis and applications, we refer to \cite{fru_chr_book_2009}. We use the latest definition of the state transition system of CHR that is based on state equivalence \cite{raiser_betz_fru_equivalence_revisited_chr09}. The definitions from those canonical sources are now reproduced. 

The syntax of CHR is defined over a set of variables, a set of function symbols with arities and a set of predicate symbols with arities that is disjointly composed of \emph{CHR constraint symbols} and \emph{built-in constraint symbols}. The set of constraint symbols contains at least the symbols $=/2$, $\top/0$ and $\bot/0$. In this paper, we allow the terms to be sets of terms as they can be simply represented as lists in implementations. For a constraint symbol $c/n$ and terms $t_1, \dots, t_n$ over the variables and function symbols, $c(t_1,\dots,t_n)$ is called a \emph{CHR constraint} or a \emph{built-in constraint}, depending on the constraint symbol. We now define the notion of CHR states.
\begin{definition}[CHR state]
A \emph{CHR state} is a tuple $\langle \mathbb{G} ; \mathbb{C} ; \mathbb{V} \rangle$ where the \emph{goal} $\mathbb{G}$ is a multi-set of constraints, the \emph{built-in constraint store} $\mathbb{C}$ is a conjunction of built-in constraints and $\mathbb{V}$ is a set of \emph{global variables}.

All variables occurring in a state that are not global are called \emph{local variables}.
\end{definition}

CHR states can be modified by rules that together form a CHR program. For the sake of brevity, we only consider simplification rules, as they are the only type of rules needed for the understanding of the paper.
\begin{definition}[CHR program]
A \emph{CHR program} is a finite set of rules of the form $r \enspace @ \enspace H \Leftrightarrow G ~|~ B_c , B_b$
 where $r$ is an optional rule name, the heads $H$ are multi-sets of CHR constraints, the guard $G$ is a conjunction of built-in constraints and the body is a multi-set of CHR constraints $B_c$ and a conjunction of built-in constraints $B_b$. If $G$ is empty, it is interpreted as the built-in constraint $\top$. 
\end{definition}

Informally, a rule is applicable, if the head matches constraints from the store $\mathbb{G}$ and the guard holds, i.e. is a consequence of the built-in constraints $\mathbb{C}$. In that case, the constraints matching $H$ are removed and the constraints from $B_c$, $B_b$ and $G$ are added.

In the context of the operational semantics, we assume a constraint theory $\mathcal{CT}$ for the interpretation of the built-in constraints. We define an equivalence relation over CHR states.
\begin{definition}[CHR state equivalence \cite{raiser_phdthesis10,raiser_betz_fru_equivalence_revisited_chr09}]
\label{def:chr_state_equiv}
Let $\rho := \langle \mathbb{G} ; \mathbb{C} ; \mathbb{V} \rangle$ and $\rho' := \langle \mathbb{G}' ; \mathbb{C}' ; \mathbb{V}' \rangle$ be CHR states with local variables $\bar{y},\bar{y}'$ that have been renamed apart. $\rho \equiv \rho'$ if and only if 
\begin{equation*}
\mathcal{CT} \models \forall(\mathbb{C} \rightarrow \exists \bar{y}'.((\mathbb{G} = \mathbb{G'}) \land \mathbb{C}')) \land \forall(\mathbb{C}' \rightarrow \exists \bar{y}.((\mathbb{G} = \mathbb{G'}) \land \mathbb{C}))
\end{equation*}
where $\forall F$ denotes the universal closure of formula $F$.
\end{definition}

The operational semantics is now defined by the following transition scheme over equivalence classes of CHR states i.e. $[\rho] := \{ \rho' ~|~ \rho' \equiv \rho \}$
\begin{definition}[operational semantics of CHR \cite{raiser_phdthesis10,raiser_betz_fru_equivalence_revisited_chr09}]
\label{def:chr_operational_semantics}
For a CHR program the state transition system over CHR states and the rule transition relation $\mapsto$ is defined as the following transition scheme:
\begin{equation*}
\frac
 {
   r \enspace @ \enspace H \Leftrightarrow G ~|~ B_c , B_b
 }
 {
 [\langle H \uplus \mathbb{G} ; G \land \mathbb{C} ; \mathbb{V} \rangle ] \mapsto^r [ \langle B_c \uplus \mathbb{G} ; G \land B_b \land \mathbb{C} ; \mathbb{V} \rangle ]
 }
\end{equation*}
Thereby, $r$ is a variant of a rule in the program such that its local variables are disjoint from the variables occurring in the representative of the pre-transition state. We may just write $\mapsto$ instead of $\mapsto^r$ if the rule $r$ is clear from the context.
\end{definition}

\subsection{Translation of ACT-R to CHR}

We briefly summarize the translation of ACT-R models to CHR first presented in \cite{gall_ruleml2016,gall_tocl_2017}. Since the translation is proven to be sound and complete \cite{gall_tocl_2017}, we explain the operational semantics of ACT-R with the help of the translation.

\begin{definition}[translation of abstract states]
\label{def:translation_states}
An abstract ACT-R state $\sigma := \langle \Delta; \gamma ; \addinfo \rangle$ can be translated to the following CHR state:
\begin{align*}
\langle & \{ \constr{delta}(\{\fun{chunk}(\id_\Delta(c),t,\llbracket \val \rrbracket) ~|~ c \in \Delta \land c = (t,\val)\}) \} \\
        & \uplus \{ \constr{gamma}(b,\id_\Delta(c),d) ~|~ b \in \buffers \land \gamma(b) = (c,d) \land c = (t,\val) \} ; \addinfo ; \emptyset \rangle 
\end{align*}
Thereby, $\llbracket \val \rrbracket$ denotes the explicit relational notation of the function $\val$ as a set of tuples. We denote the translation of an ACT-R state $\sigma$ by $\fun{chr}(\sigma)$.
\end{definition}

The chunk store is represented by a $\constr{delta}$ constraint that contains a set of $\fun{chunk}/3$ terms representing the chunks with their identifiers, types and slot-value pairs. 

For every buffer of the given architecture, there is a constraint $\constr{gamma}$ with buffer name, chunk identifier and delay. Since $\gamma$ is a total function, every buffer has exactly one $\constr{gamma}$ constraint. Additional information is represented directly as built-in constraints.

\begin{definition}[translation of rules]
\label{def:translation_rules}
Let $\fun{cogstate}(\buffers) := \{ (b,C_b) ~|~ b \in \buffers \}$ be the relation that connects each buffer with a variable $C_b$. An ACT-R rule in set-normal form $r := \mathcal{L} \Rightarrow \mathcal{R}$ can be translated to a CHR rule of the form:
\begin{align*}
r ~@~  &  \constr{delta}(D) \uplus \{ \constr{gamma}(b,C_b,E_b) ~|~ b \in \buffers \} \\
       \Leftrightarrow \\
       &\bigwedge_{\test{b}{t}{P} \in \mathcal{L}} ( \fun{chunk}(C_b,t,P) ~\constr{in}~ D \land E_b {=} 0 ) ~|~ \\
        & \{ \constr{delta}(D^*)\}  \uplus \{ \constr{gamma}(b,C_b^{**},\fun{resdelay}(b)) ~|~ b \in \buffers \land a(b,t,P) \in R \} \\
        \uplus~    &  \{ \constr{gamma}(b,C_b,E_b) ~|~ b \in \buffers \land a(b,t,P) \notin R \} ,\\
                  & \bigwedge_{\alpha = a(b,t,P) \in \mathcal{R}}  \constr{action}(\alpha, D, \fun{cogstate}(\buffers), D_b^*, C_b^*, E_b^*) \\
       \land~      & \constr{merge}([D_b^* : a(b,t,P) \in R],D') \} \land \constr{merge}([D,D'],D^*) ~\uplus\\
       \land~      & \bigwedge_{a(b,t,P) \in R} \constr{map}(D,D',C_b^*,C_b^{**}).
\end{align*}
Note that ACT-R constants and variables from $\consts$ and $\vars$ are implicitly translated to corresponding CHR variables. 

We denote the translation of a rule $r$ by $\fun{chr}(r)$ and the translation of an ACT-R model $\Sigma$ that is a set of ACT-R rules by $\fun{chr}(\rules)$. Thereby, $\fun{chr}(\rules) := \{ \fun{chr}(r) ~|~ r \in \rules \}.$
\end{definition}
The rule removes the $\constr{delta}$ and all $\constr{gamma}$ constraints from the store. It binds the translation of the chunk store $\Delta$ to the variable $D$. For all buffers $b$, each variable $C_b$ is bound to the chunk identifier of the chunk in $b$, i.e. $C_b = \id_\Delta(\gamma(b))$. The guard now performs all buffer tests $\test{b}{t}{P}$ from the ACT-R rule by testing if a chunk term $\fun{chunk}(C_b,t,P)$ is in the translated chunk store $D$ that has type $t$ and matches all slot-value pairs in $P$. The ACT-R variables in $P$ are bound to the values in the state.

In the body, the built-in constraints $\constr{action}$ perform the actions of the ACT-R rule as defined by the architecture. 
An $\constr{action}$ constraint gets the action term $\alpha$ of the rule (with all variables bound through the matching), the original chunk store and a representation of the cognitive state. Since the $C_b$ have been bound in the matching, it consists of tuples that connect each buffer $b$ with the chunk identifier it holds.

The $\constr{action}$ built-in constraint returns a chunk store $D_b^*$, a chunk identifier $C_b^*$ that represents the resulting chunk from the request and a result delay $E_b^*$. The $\constr{merge}$ constraints merge the chunk stores of all actions with the original store $D$ to the store $D^*$. The result of merging two chunk stores can vary from implementation to implementation, but has to obey some rules defined in \cite{gall_tocl_2017}. One can think of it as a multi-set union. As chunk identifiers might change in the merging process, the built-in $\constr{map}$ maps the chunk identifier of the results to the corresponding identifiers in the merged store.

Then, a new $\Delta$ constraint with the resulting chunk store $D^*$ is added as well as the $\constr{gamma}$ constraints. If the buffer $b$ has been part of an action, then it is altered such that it holds the resulting chunk identifier $C_b^{**}$ after the merge and the resulting delay $E_b^*$. If it was not part of an action, its parameters $C_b$ (the chunk identifier) and $E_b$ (the delay) remain unchanged. This is possible, since the chunk merging guarantees that chunks in the original chunk store $D$ the constraint $\constr{gamma}$ is referring to, are also part of the merged chunk store $D^*$.

\begin{example}[counting]
\label{ex:counting}
We now give an example ACT-R rule to explain its operational semantics.
A classical example in ACT-R is counting by recalling count order facts. The model uses chunks of type \emph{order} as illustrated in figure~\ref{fig:chunk_count_order}. An \emph{order} chunk has a \emph{first} and a \emph{second} slot that link two chunks representing natural numbers in the right order. Additionally, we define a second chunk type \emph{g} that memorizes the current number in the counting process. The main rule is defined as:
\begin{align*}
 & \test{goal}{g}{\{(current,X)\}}, \test{retrieval}{order}{\{(first,X), (second,Y)\}} \\
\Rightarrow \enspace & \actionmod{goal}{g}{\{(current,Y)\}}, \actionreq{retrieval}{order}{\{(first,Y)\}}
\end{align*}
The left-hand side tests if there is a chunk of type \emph{g} in the \emph{goal} buffer. The value of its \emph{current} slot is bound to variable $X$ by the matching. The second buffer test checks the \emph{retrieval} buffer for a chunk of type \emph{order} that has $X$ in its \emph{first} slot. The value of the \emph{second} slot is bound to variable $Y$.

The right-hand side modifies the chunk in the \emph{goal} buffer such that $Y$ is written to the \emph{current} slot. The second action requests the \emph{retrieval} buffer for an \emph{order} chunk that has $Y$ in its \emph{first} slot. As soon as the requested chunk is available, the program can apply the rule again. The head and guard of the CHR translation $H \Leftrightarrow G ~|~ B$ of the rule is 
\begin{align*}
 H & := \{ \constr{delta}(D), \constr{gamma}(g,C_g,0), \constr{gamma}(\fun{retrieval},C_r,0) \},\\
 G & := \fun{chunk}(C_g, g, \{ (\fun{current}, X) \}) ~\constr{in}~ D ~\land \\
   &    \fun{chunk}(C_r, \fun{order}, \{ (\fun{first}, X), (\fun{second},Y) \}) ~\constr{in}~ D.
\end{align*}
\end{example}

%

\section{Confluence Criterion for ACT-R}
\label{sec:confluence_criterion}

This section is the main contribution of the paper. We gradually develop a decidable criterion for confluence of ACT-R using the CHR embedding.

Therefor, a brief introduction to invariant-based confluence analysis for CHR is given that extends the standard confluence criterion to handle invariants that must hold for the regarded states. We then define the ACT-R invariant $\actrinv$ that is satisfied if a CHR state has been derived from an ACT-R state. Then a decidable criterion for the invariant is presented and it is shown that the invariant is maintained in translated ACT-R models. It is shown how invariant-based confluence analysis for CHR can be applied to decide ACT-R confluence.

\subsection{Invariant-based Confluence}
\label{sec:confluence_criterion:invariant_confluence}

We now give a brief introduction to invariant-based confluence analysis for CHR. The first results stem from \cite{duck_observable_2007}. We summarize the main theorem of the improved version that can be found in \cite[section~14]{raiser_phdthesis10}.

The main idea of the confluence criterion is that heads and guards of the rules are overlapped to an overlap state. Then both overlapping rules are applied to this state forming a critical pair that is checked for joinability for all possible overlap states. An overlap is defined as follows:
\begin{definition}[overlap and critical pairs \cite{fru_chr_book_2009,raiser_phdthesis10}]
For any two (not necessarily different) rules of a CHR program with renamed apart variables of the form $r ~@~ H \Leftrightarrow G ~|~ B_c, B_b$ and $r' ~@~ H' \Leftrightarrow G' ~|~ B_c', B_b'$, let $O \subseteq H$, $O' \subseteq H'$ such that for $B := (O = O') \land G \land G'$ it holds that $\mathcal{CT} \models \exists.B$ and $O \neq \emptyset$, then the state 
\begin{equation*}
 \sigma = \langle R \uplus R' \uplus O ; B ; \mathbb{V} \rangle
\end{equation*}
is called an \emph{overlap} of $r$ and $r'$ where $R := H \setminus O$, $R' := H' \setminus O'$ and $\mathbb{V}$ is the set of all variables occurring in heads and guards of both rules. The pair of states
$\sigma_1 := \langle R' \uplus B_c ; B \land B_b ; \mathbb{V} \rangle$ and $\sigma_2 := \langle R \uplus B_c' ; B \land B_b' ; \mathbb{V} \rangle$
is a \emph{critical pair} of the overlap $\sigma$.

\end{definition}

CHR has the  monotonicity property. It states that all rules that are applicable in a state, are also applicable in any larger state. This idea can be exploited to reason from joinable overlap states about local confluence and therefore confluence of a CHR program. The problem with invariant-based confluence is that the idea of using monotonicity to reason about larger states does not work for states where the invariant does not hold. An overlap that does not satisfy the invariant makes all information about this state irrelevant \cite[p. 79]{raiser_phdthesis10}. The idea of the invariant-based confluence theorem for CHR is to extend all states where the invariant does not hold such that the invariant is repaired and include the extended states in the confluence test. Since in general there are infinitely many extensions that maintain the invariant, only minimal extensions according to a partial order defined in \cite{raiser_phdthesis10} have to be considered. Then, monotonicity can be applied again. 
\begin{theorem}[invariant-based confluence for CHR \cite{raiser_phdthesis10}]
\label{thm:inv_confluence} 
 For an invariant $\inv$, let $\Sigma^\inv([\rho]) := \{ [\rho'] ~|~ \text{ $[\rho']$ is an extension of $[\rho]$ such that $\inv$ holds } \}$ be the set of \emph{satisfying extensions} of $[\rho]$. The set $\mathcal{M}^\inv([\rho])$ is the set of \emph{minimal elements} of $\Sigma^\inv([\rho])$ w.r.t. the partial order on states defined in \cite{raiser_phdthesis10}.

Let $\mathcal{P}$ be a CHR program and $\mathcal{M}^\inv([\rho])$ be well-defined for all overlaps $\rho$. $\mathcal{P}$ is locally confluent with respect to $\inv$ if and only if for all overlaps $\rho$ with critical pairs $(\rho_1, \rho_2)$ and all $[\rho_\mathrm{m}] \in \mathcal{M}^\inv([\rho])$ holds that $[\rho_1]$ extended by $[\rho_\mathrm{m}]$ and $[\rho_2]$ extended by $[\rho_\mathrm{m}]$ are joinable. We then say that $\mathcal{P}$ is \emph{$\inv$-(locally) confluent}.
\end{theorem}
There are two problems with this result making it possibly undecidable: The invariant could be undecidable and the set of minimal elements can be infinitely large. We will show that in the case of the ACT-R invariant that we use for our confluence test, the set of satisfying extensions is empty and the invariant is decidable. Hence, it is not necessary for the understanding of this paper how the partial order on states and therefore the set of minimal elements is defined formally, since the set of satisfying extensions is already empty for the ACT-R invariant. The ACT-R invariant is defined in the following section.

\subsection{ACT-R Invariant}
\label{sec:confluence_criterion:actr_invariant}

To reason about confluence of ACT-R models in CHR, we need an invariant that restricts the CHR state space to states that stem from a valid ACT-R state. In the following example, we show how overlapping translated ACT-R rules can lead to overlap states that do not describe a valid ACT-R state.
\begin{example}
Let $\{ \constr{delta}(D), \constr{gamma}(B,C,0) \} \Leftrightarrow \fun{chunk}(C,T,P) ~\constr{in}~ D ~|~ \dots$ be a CHR rule that has been obtained from an ACT-R rule. By overlapping the rule with itself, we could get 
\begin{align*}
 \sigma := & \langle \constr{delta}(D), \constr{gamma}(B,C,0), \constr{gamma}(B,C',0) ;\\
 &\fun{chunk}(C,T,P) ~\constr{in}~ D \land \fun{chunk}(C',T',P') ~\constr{in}~ D ; \mathbb{V} \rangle.
\end{align*}
However, this state does not stem from a valid ACT-R state, since $\gamma$ is a function with only one value for each buffer and therefore the translation of an ACT-R state can never contain two $\constr{gamma}$ constraints for the same buffer $B$.
\end{example}

In the following, we define the ACT-R invariant $\actrinv$ on CHR states that limits the state space to states that stem from valid ACT-R states. We show that the invariant is decidable by breaking it down to five fine grained invariants. We also show that it actually defines an invariant for translated ACT-R models.

\begin{definition}[ACT-R invariant]
Let $[\rho]$ be a CHR state. The \emph{ACT-R invariant} $\actrinv$ holds if and only if there is an ACT-R state $\sigma$ such that $\rho \equiv \fun{chr}(\sigma)$.
\end{definition}
Basically, this means that $\actrinv([\rho])$ holds if $[\rho]$ is the valid translation of an ACT-R state. However, by this definition it is hard to decide if a CHR state satisfies the invariant.

We now show some decidable sub-invariants on CHR states and prove that their conjunction is equivalent to $\actrinv$. For this purpose, we define an auxiliary function $\fun{ids}$ that returns the set of chunk identifiers for a set of $\fun{chunk}/3$ terms.
\begin{definition}[chunk identifiers]
Let $d$ be a set. Then 
\begin{equation*}
\fun{ids}(d) := \{ c ~|~ \fun{chunk}(c,t,p) \in d \}
\end{equation*}
is the \emph{set of chunk identifiers} of the set $d$.
\end{definition}

The sub-invariants mainly consist of \emph{uniqueness} invariants, i.e. they require that there is only one constraint of a certain kind for a class of arguments, and \emph{functional dependency} invariants, i.e. that certain sets that represent relations appearing in constraints are functions. Eventually, the constraints that can be be used in a state are restricted.

\begin{theorem}[ACT-R invariants]
\label{def:invariants}
Let $\rho \equiv \langle \mathbb{G} ; \mathbb{C} ; \mathbb{V} \rangle$ be a CHR state. We define the following sub-invariants:
\begin{enumerate}
  \item unique chunk store
 \label{def:invariants:unique_chunk_store}
 
 $\actrinv_{\ref{def:invariants:unique_chunk_store}}([\rho]) \leftrightarrow$ 
  There is exactly one constraint $\constr{delta}(d) \in \mathbb{G}$ for some ground set $d$. For all elements $e \in d$, it holds that there exist $c \in \consts, t \in \types, p \in \consts \times \consts, s \in \tau(t), v \in \consts$ such that $e = \fun{chunk}(c,t,p)$ and $p = \{ (s,v) ~|~ s \in \tau(t) \land v \in \fun{ids}(d) \}$.

 \item functional dependency of cognitive state 
  \label{def:invariants:func_cogstate}
 
$\actrinv_{\ref{def:invariants:func_cogstate}}([\rho]) \leftrightarrow$ 
 For all buffers $b \in \buffers$ there is exactly one $\constr{gamma}(b,c,e) \in \mathbb{G}$ where $c \in \fun{ids}(d)$ for some $\constr{delta}(d) \in \mathbb{G}$ and $e \in \mathbb{R}_0^+$.
 
 \item unique chunk identifiers
  \label{def:invariants:unique_chunk_ids}
 
 $\actrinv_{\ref{def:invariants:unique_chunk_ids}}([\rho]) \leftrightarrow$ 
 For all chunk identifiers $c \in \consts$ and constraints $\constr{delta}(d) \in \mathbb{G}$, if $\fun{chunk}(c,t,p) \in d$, then there is no other term $\fun{chunk}(c,t',p') \in d$.

 \item functional dependency of slot-value pairs
 \label{def:invariants:func_svp}
 
 $\actrinv_{\ref{def:invariants:func_svp}}([\rho]) \leftrightarrow$ 
 For all constraints $\constr{delta}(d) \in \mathbb{G}$, terms $\fun{chunk}(c,t,p)$ in set $d$ and $(s,v)$ in set $p$, there is no other term $(s,v')$ in $p$.
 
 \item allowed constraints
 \label{def:invariants:allowed_constr}
 
 $\actrinv_{\ref{def:invariants:allowed_constr}}([\rho]) \leftrightarrow$ 
 In $\mathbb{G}$ there are only $\constr{delta}/1$ and $\constr{gamma}/3$ constraints, only syntactic equality $=/2$ and the allowed constraints defined by the ACT-R architectures appear in $\mathbb{C}$ and $[\rho]$ is ground.
\end{enumerate}

For all CHR states $[\rho]$ it holds that $\actrinv([\rho]) \leftrightarrow \bigwedge_{i=1}^{5} \actrinv_i([\rho]).$
\end{theorem}
\begin{proof}
\begin{description}
 \item[if direction] 
 
 If $\actrinv([\rho])$, then $[\rho]$ is the product of the translation of an ACT-R state. It follows directly from definition~\ref{def:translation_states} that in that case, $\actrinv_1([\rho])$, $\actrinv_2([\rho])$, $\actrinv_3([\rho])$, $\actrinv_4([\rho])$ and $\actrinv_5([\rho])$ hold.
 
 \item[only-if direction] 
 
 We have to show that for all CHR states $[\rho]$ where the invariants $\actrinv_1([\rho])$, $\actrinv_2([\rho])$, $\actrinv_3([\rho])$, $\actrinv_4([\rho])$ and $\actrinv_5([\rho])$ hold, there is an ACT-R state $\sigma$ such that $\rho \equiv \fun{chr}(\sigma)$. Let $[\rho] := [\langle \mathbb{G} ; \mathbb{C} ; \mathbb{V} \rangle]$.
 
 We construct the ACT-R state $\sigma := \langle \Delta ; \gamma ; \addinfo \rangle$. Since $\actrinv_{\ref{def:invariants:unique_chunk_store}}([\rho])$, there is exactly one $\constr{delta}(d)$ constraint for a set $d$ and all elements in $d$ are of the form $\fun{chunk}(c,t,p)$ where $c \in \consts, t \in \types$ and $p$ is a set of elements $(s,v)$ with $s \in \tau(t)$ and $v \in \fun{ids}(d)$. The set $p$ is total with respect to $s$ and the $v$ are chunk identifiers that appear in $d$. Due to $\actrinv_{\ref{def:invariants:func_svp}}$, there is exactly one $(s,v) \in p$ for each $s \in \tau(t)$, hence $p$ is the relational representation of a value function.The invariant $\actrinv_{\ref{def:invariants:unique_chunk_ids}}$ guarantees that the chunk identifiers are unique. 
 
 We define $\Delta := \{ (t,p) ~|~ \fun{chunk}(c,t,p) \in d \}$ with the identifier function $\id_\Delta := \{ ((t,p),c) ~|~ \fun{chunk}(c,t,p) \}$.
 
 Due to invariant $\actrinv_{\ref{def:invariants:func_cogstate}}$, the cognitive state can then be defined for all $b \in \buffers$ such that $\gamma(b) := (\id_\Delta^{-1}(c),e)$ for each $\constr{gamma}(b,c,e) \in \mathbb{G}$.

 Since $\actrinv_{\ref{def:invariants:allowed_constr}}([\rho])$, $[\rho]$ is ground. Hence, we can find another representative of the state with $\rho \equiv \langle \mathbb{G}' ; \mathbb{C}' ; \emptyset \rangle$, that applies all equality constraints $X {=} t$ in $\mathbb{C}$ such that only constants appear in $\mathbb{G}'$ and $\mathbb{C}'$ and $\mathbb{C}'$ only consists of allowed predicates defined by the ACT-R architecture. Therefore, we can set $\addinfo := \mathbb{C}'$.
 
 From the construction of $\sigma$ it is clear that $\rho \equiv \fun{chr}(\sigma)$.
\end{description}

\end{proof}

The invariants $\actrinv_1, \dots, \actrinv_5$ are obviously decidable. Since they are equivalent to the ACT-R invariant $\actrinv$, theorem~\ref{def:invariants} gives us a decidable criterion for the ACT-R invariant $\actrinv$.

In the next step, we show that the ACT-R invariant $\actrinv$ is maintained by transitions that come from a translated ACT-R program, i.e. that it really is an invariant.
\begin{lemma}
\label{lemma:invariant_maintained}
Let $\mapsto$ be the state transition relation derived from the translation of an ACT-R model and $[\rho]$ a CHR state with $\actrinv([\rho])$. If $[\rho] \mapsto [\rho']$, then $\actrinv([\rho'])$.

\end{lemma}
\begin{proof}
We are going to use soundness and completeness \cite{gall_tocl_2017} to prove this.

Let $[\rho]$ be a CHR state with $\actrinv([\rho])$. Since $\actrinv([\rho])$, there is an ACT-R state $\sigma$ with $\rho \equiv \fun{chr}(\sigma)$. Due to the sound and complete embedding of ACT-R in CHR, there is an ACT-R state $\sigma'$ with $\rho' \equiv \fun{chr}(\sigma')$. Hence, $\actrinv([\rho'])$ holds.
\end{proof}

\subsection{Invariant-Based Confluence Test}
\label{sec:confluence_criterion:invariant_confluence_actr}

We want to use theorem~\ref{thm:inv_confluence} \cite[p. 83, theorem~6]{raiser_phdthesis10} to prove confluence of all states $[\rho]$ that satisfy the ACT-R invariant, i.e. where $\actrinv([\rho])$. Therefore, we have to construct the set $\Sigma^\actrinv([\rho])$ for each state $[\rho]$ that does not satisfy $\actrinv$. It contains all states that can be merged to $[\rho]$ such that they satisfy $\actrinv$ (see theorem~\ref{thm:inv_confluence}). The minimal elements in this set have to be considered in the confluence test.

We will see that for all states $[\rho]$ that do not satisfy $\actrinv$, the set of minimal elements is empty. Intuitively, this means that there are no states that can extend $[\rho]$ such that it satisfies $\actrinv$.
\begin{lemma}[minimal elements for $\actrinv$]
\label{lemma:min_elements}
 Let $\actrinv$ be the ACT-R invariant as defined in definition~\ref{def:invariants}. For all states $[\rho]$ such that $\actrinv([\rho])$ does not hold, $\Sigma^\actrinv([\rho]) = \emptyset$ and therefore $\mathcal{M}^\actrinv([\rho]) = \emptyset$.
\end{lemma}
\begin{proof}
Let $[\rho] := [\langle \mathbb{G} ; \mathbb{C} ; \mathbb{V} \rangle]$. We use theorem~\ref{def:invariants} that allows us to analyze the individual sub-invariants:
\begin{enumerate}
 \item If $\actrinv_{\ref{def:invariants:unique_chunk_store}}$ is violated, there are the following cases:
 \begin{itemize}
  \item There are two constraints $\constr{delta}(d), \constr{delta}(d') \in \mathbb{G}$. We cannot extend $[\rho]$ (i.e. add constraints) to satisfy $\actrinv_{\ref{def:invariants:unique_chunk_store}}$.
  \item There is only one unique $\constr{delta}(d) \in \mathbb{G}$, with elements that do not have the required form. Again, no constraints can be added to satisfy $\actrinv_{\ref{def:invariants:unique_chunk_store}}$.
 \end{itemize}
 \item If $\actrinv_{\ref{def:invariants:func_cogstate}}$ is violated, there are two constraints $\constr{gamma}(b,c,e), \constr{gamma}(b',c',e') \in \mathbb{G}$. We cannot satisfy $\actrinv_{\ref{def:invariants:func_cogstate}}$ for such a state.
 \item The proof is analogous for $\actrinv_{\ref{def:invariants:unique_chunk_ids}}$ and $\actrinv_{\ref{def:invariants:func_svp}}$.
 \item If $\actrinv_{\ref{def:invariants:allowed_constr}}$ is violated, there are other constraints then $\constr{delta}$ or $\constr{gamma}$ in $\mathbb{G}$ or other than the allowed constraints defined by the architecture in $\mathbb{C}$. This cannot be repaired by extending $\mathbb{G}$ or $\mathbb{C}$.
\end{enumerate}

\end{proof}

We can directly apply theorem~\ref{thm:inv_confluence}: For all overlaps $\rho$ where $\actrinv([\rho])$ holds, the set of minimal elements is $\mathcal{M}^\actrinv([\rho]) = \{ [\rho_\emptyset] \}$ \cite[p.80, lemma 13.13]{raiser_phdthesis10} where $\rho_\emptyset := \langle \emptyset ; \top ; \emptyset \rangle$ is the empty CHR state. Hence, for overlaps where $\actrinv$ holds, we only have to show joinability of the critical pairs that stem from the overlap itself. This coincides with the regular confluence test of CHR as defined in \cite{fru_chr_book_2009}.

For all overlaps $\rho$ where $\actrinv([\rho])$ does not hold, the set of minimal elements is $\mathcal{M}^\actrinv([\rho]) = \emptyset$ by lemma~\ref{lemma:min_elements}. Therefore, no critical pairs have to be tested. We summarize this in the following theorem.
\begin{theorem}[$\actrinv$-local confluence]
\label{thm:actr_local_confluence}
A CHR program is $\actrinv$-local confluent if and only if for all critical pairs $(\rho_1,\rho_2)$ with overlap $\rho$ for which $\actrinv(\rho)$, it is $\rho_1 \downarrow \rho_2$.
\end{theorem}
\begin{proof}
This follows directly from theorem~\ref{thm:inv_confluence} and lemma~\ref{lemma:min_elements} for overlaps where $\actrinv([\rho])$ does not hold. For overlaps with $\actrinv([\rho])$, the unique minimal element is the empty state $[\rho_\emptyset] := [\langle \emptyset ; \top ; \emptyset \rangle]$ which is the neutral element for state merging \cite[lemma~13.13, p.~80]{raiser_phdthesis10}. Therefore, if $\actrinv([\rho])$ holds, it suffices to test the critical pairs that stem from $[\rho]$ by theorem~\ref{thm:inv_confluence}.
\end{proof}

We now have a criterion to decide $\actrinv$-confluence of $\actrinv$-terminating CHR programs that have been translated from an ACT-R model. In the next theorem, we show that $\actrinv$-confluence of such CHR programs coincides with ACT-R confluence. Therefore, the confluence criterion is applicable to decide confluence of ACT-R models.

\begin{theorem}[confluence in ACT-R]
\label{thm:actr_confl_eq_chr_confl}
Let $M$ be an ACT-R model. Then $M$ is terminating and confluent if and only if $\fun{chr}(M)$ is $\actrinv$-terminating and $\actrinv$-confluent.
\end{theorem}
\begin{proof} 
$\actrinv$-termination is maintained through soundness and completeness. We now show that confluence for terminating models and their CHR counterparts coincides. Confluence is defined as $(\sigma \mapsto^* \sigma_1) \land (\sigma \mapsto^* \sigma_2) \rightarrow (\sigma_1 \downarrow \sigma_2)$
for all states $\sigma, \sigma_1, \sigma_2$. It remains to show that joinability in ACT-R and CHR are equivalent, i.e. $(\sigma_1 \downarrow \sigma_2) \leftrightarrow ([\fun{chr}(\sigma_1)] \downarrow [\fun{chr}(\sigma_2)]).$

\begin{description}
 \item[If-direction] 
 If $(\sigma_1 \downarrow \sigma_2)$, there is a state $\sigma'$ such that $\sigma_1 \utrans^* \sigma'$ and $\sigma_2 \utrans^* \sigma'$. Due to soundness and completeness of the embedding, we have that $[\fun{chr}(\sigma_1)] \mapsto^* [\fun{chr}(\sigma')]$ and $[\fun{chr}(\sigma_2)] \mapsto^* [\fun{chr}(\sigma')]$. 

  \item[Only-if-direction] 
 This is analogous. We just have to construct the ACT-R state from the joined CHR state $[\rho']$. Since $\actrinv([\rho'])$ holds by lemma~\ref{lemma:invariant_maintained}, this state exists.
\end{description}
\end{proof}

\subsection{Example: Counting}
\label{sec:example:counting}

We continue our example~\ref{ex:counting}. We assume that each number chunk only appears in at most one \emph{order} chunk at \emph{first} or \emph{second} position. This means that the model has learned a stable order on the numbers and hence requests to the declarative module are deterministic.
It is clear that this example model terminates for finite declarative memories. Therefore, we can apply our confluence criterion. 

The rule can overlap with itself, e.g. $\langle \constr{delta}(D), \constr{delta}(D'), \ldots ; \ldots ; \ldots  \rangle.$ This state invalidates invariant $\actrinv_{\ref{def:invariants:unique_chunk_store}}$ and hence is not part of the confluence test. Another overlap is $\langle \constr{delta}(D), \constr{gamma}(g,C_g,0), \constr{gamma}(g,C_g',0), \dots ; \dots ;  \dots \rangle.$
It violates invariant $\actrinv_{\ref{def:invariants:func_cogstate}}$, because it has two $\constr{gamma}$ constraints for the same buffer. 

All overlaps consist of the following built-in store:
\begin{align*}
\langle  \constr{delta}(D), \dots ; &\fun{chunk}(C_g, g, \{ (\fun{current}, X) \}) ~\constr{in}~ D \\
                             \land~ &\fun{chunk}(C_g, g, \{ (\fun{current}, X') \}) ~\constr{in}~ D, \dots ; \{ D, X, X', \dots \} \rangle.
\end{align*}
By invariant $\actrinv_{\ref{def:invariants:unique_chunk_ids}}$ it must be $X = X'$, because otherwise there were two different $\fun{chunk}$ terms in the same chunk store with the same chunk identifier.

The overlap $\langle H ; G ; \mathbb{V}\rangle$ that only consists of the head and guard of the rule where $\mathbb{V}$ contains all variables of $H$ and $G$ is joinable, because we assumed determinism of requests, i.e. there is only one possible result chunk for each request. It can be seen that all possible overlaps in this small example invalidate the ACT-R invariant $\actrinv$ or are joinable. Therefore, the model consisting only of this one counting rule is confluent. If we would assume an agent that has not learned a stable order of numbers, yet, i.e. there are numbers with different successors, the model would not be confluent. The confluence test constructs minimal representations of the states that are not joinable, i.e. giving an insight to the reason why a model is not confluent. This allows to decide whether the model has the desired behavior when it comes to different available strategies.

\section{Related Work}

There exist CHR embeddings of other rule-based approaches. The results on invariant-based confluence analysis have been used successfully to the embedding of graph transformation systems in CHR \cite{raiser_gts_chr_2007,raiser_analysis_gts_2011}.

In the context of ACT-R, there are -- to the best of our knowledge -- no other approaches that deal with confluence so far. There have been other approaches to formalize the architecture with the aim to reason about cognitive models. For instance, F-ACT-R \cite{albrecht_2014a,albrecht_2014c} formalizes the architecture of ACT-R to simplify comparison of different models or to use model checking techniques. In \cite{said2016applying} mathematical reformulations of ACT-R models are used for parameter optimization by mathematical optimization techniques.

\section{Conclusion}

In this paper, we have shown a decidable confluence test for the abstract operational semantics of ACT-R. A confluence test can help to improve ACT-R models by identifying the rules that inhibit confluence. This enables the modeler to decide about the correct behavior of the model regarding competing strategies. In our approach, we use the sound and complete embedding of ACT-R in CHR to apply the invariant-based confluence criterion for CHR to reason about ACT-R confluence, since standard CHR confluence is too strict. 

We have defined the ACT-R invariant $\actrinv$ on CHR states such that it is satisfied for all states that stem from a valid ACT-R state. The first main result is a decidable criterion for the ACT-R invariant (theorem~\ref{def:invariants}). 

Furthermore, the theoretical foundations for applicability of CHR invariant-based confluence for the ACT-R invariant $\actrinv$ are established. This leads to the second main result: an invariant-based CHR $\actrinv$-confluence test (theorem~\ref{thm:actr_local_confluence}). 

Eventually, it is shown that $\actrinv$-confluence coincides with ACT-R confluence (theorem~\ref{thm:actr_confl_eq_chr_confl}). This makes our CHR approach applicable to decide ACT-R confluence. The criterion is decidable as long as the constraint theories behind the actions are decidable, because the invariant is decidable and the preconditions for the invariant-based confluence test are satisfied in the context of ACT-R. 

For the future, we want to investigate how the approach can be extended to confluence modulo equivalence \cite{christiansen_confmeq_2016}, since ACT-R confluence can be too strict due to possibly differing chunk identifiers in the processing of the production rules. An equivalence relation on chunk networks that is defined as a special form of graph isomorphism could abstract from chunk identifiers making a chunk store more declarative. By summarizing possible outcomes of a model in equivalence classes, confluence modulo equivalence can also help to reason about correctness of a model. Confluence modulo this equivalence relation would then guarantee that the model always gives a result of a certain kind defined by the equivalence class. For instance, it would be possible to check if a model always yields a chunk of a certain type, e.g. a number or an order chunk.

Reasoning about requests to modules that appear in a confluence proof can be extended by specific constraint theories on the modules that integrate domain-specific knowledge about the model. This idea can be extended by allowing for model-specific constraint theories. For instance, the integration of domain-specific knowledge on chunk types in the context of a particular cognitive model could improve reasoning about module requests in such models.

%
%

\bibliographystyle{splncs03}
\bibliography{bib}



\end{document}